\pdfoutput=1

\documentclass[11pt]{article}

\usepackage{appendix}

\usepackage[figuresright]{rotating}
\usepackage[applemac]{inputenc}

\usepackage{tikz}
\usepackage{pgfplots}

\usepgfplotslibrary{groupplots} 
\usetikzlibrary{pgfplots.groupplots}

\usepackage{graphicx}
\usepackage{dcolumn}
\usepackage{bm}
\usepackage{amscd,amsthm}
\usepackage{ifthen}
\usepackage{booktabs}

\usepackage{hyperref}
\hypersetup{
    bookmarks=true,         
    unicode=false,          
    pdftoolbar=true,        
    pdfmenubar=true,        
    pdffitwindow=false,     
    pdfstartview={FitH},    
    pdfsubject={Subject},   
    pdfproducer={Producer}, 
    pdfnewwindow=true,      
    colorlinks=true,       
    linkcolor=red,          
    citecolor=green,        
    filecolor=magenta,      
    urlcolor=cyan           
}

\newtheorem{lemma}{Lemma}

\newtheorem{theorem}{Theorem}

\usetikzlibrary{matrix,arrows,decorations.pathmorphing}

\usepackage[
backend=bibtex,
url=false,isbn=false,doi=false,
hyperref=auto,
style=alphabetic,
natbib=true,
sorting=nty,
firstinits=true,
maxbibnames=10,
]{biblatex}

\bibliography{./bibliography.bib} 

\definecolor{DarkBlue}{rgb}{0,0,0.7} 

\usepackage{tikz,pgf,pgfplots}
\usetikzlibrary{calc,shadows}
\usetikzlibrary{pgfplots.groupplots}

\usepackage{times}
\usepackage{xr}

\usepackage{amsthm}
\usepackage{amsmath}    
\usepackage{amssymb}    
\usepackage{mathrsfs}   
\usepackage{epsfig}     
\usepackage{graphicx}
\usepackage{url}

\usepackage{url}
\usepackage{color}
\usepackage{booktabs}

\usepackage{subfigure}

\usepackage{enumitem}

\graphicspath{{}{figs/}{../}{../figs/}{../Monograph/monograph/nowformat/}}

\renewcommand{\caption}[1]{\singlespacing\hangcaption{#1}\normalspacing}

\usepackage{amsmath}
\usepackage{amsfonts}
\usepackage{dsfont}
\usepackage{mathrsfs}

\usepackage{algorithmic}
\usepackage{algorithm}
\usepackage{multirow}

\newcommand{\Z}{{\mathds Z}}

\newcommand{\bA}{\mathsf{A}}
\newcommand{\bC}{\mathsf{C}}
\newcommand{\bG}{\mathsf{G}}
\newcommand{\bT}{\mathsf{T}}

\newcommand{\C}{{\mathcal C}}

\newcommand{\E}{{\mathcal E}}

\newcommand{\M}{{\mathcal M}}
\newcommand{\N}{{\mathcal N}}

\newcommand{\T}{{\mathcal T}}

\newcommand{\x}{{\bf x}}
\newcommand{\y}{{\bf y}}

\newtheorem{cor}{Corollary} 

\newcommand{\blfootnote}[1]{%
  \begingroup
  \renewcommand\thefootnote{}\footnote{#1}%
  \addtocounter{footnote}{-1}%
  \endgroup
}

\newcommand{\one}{\mathds 1}

\newcounter{constcount}
\newcounter{numcount}
\newcounter{thmcnt}
  \let\Oldsection\section
  
\renewcommand{\section}{\stepcounter{thmcnt}\Oldsection}

\allowdisplaybreaks

\newcommand{\aln}[1]{\begin{align*}#1\end{align*}}

\newcommand{\al}[1]{\begin{align}#1\end{align}}

\def\Item$#1${\item $\displaystyle#1$
   \hfill\refstepcounter{equation}(\theequation)}

\newcommand{\bea}{\begin{eqnarray}}
\newcommand{\eea}{\end{eqnarray}}
\newcommand{\beas}{\begin{eqnarray*}}
\newcommand{\eeas}{\end{eqnarray*}}

\usepackage{ifthen}
\usepackage{mathtools}

\newcommand\Tex{}
\newcommand\PR[2][\Tex]{
\ifthenelse{\equal{#1}{}}{{\mathrm{Pr}}\left(#2\right)}{\ensuremath{{\mathrm{Pr}}_{#1}\left[ #2\right]}}}

\newcommand\EX[2][\Tex]{
\ifthenelse{\equal{#1}{}}{{\mathbb E}\left[#2\right]}{\ensuremath{{\mathbb E}_{#1}\left[ #2\right]}}}

\newcommand\Var[2][\Tex]{
\ifthenelse{\equal{#1}{}}{{\mathrm{Var}}\left[#2\right]}{\ensuremath{{\mathrm{Var}}_{#1}\left[ #2\right]}}}
\newcommand\defeq{\coloneqq}

\renewcommand\M{M} 
\renewcommand\N{N} 
\newcommand\len{L} 

\newcommand\vf{\mathbf{f}} 

\newcommand\cNM{c} 

\newcommand{\type}{t}

\begin{document}

\title{}

\begin{center}

{\bf{\LARGE{
Fundamental Limits of DNA Storage Systems
}}}

\vspace*{.2in}

{\large{
\begin{tabular}{cccc}
Reinhard Heckel$^{\dagger,\ast}$ & Ilan Shomorony$^{\dagger,\ast}$ & Kannan
Ramchandran$^{\dagger}$ & David N.~C.~Tse$^{\ddagger}$ \\
\end{tabular}
}}

\vspace*{.2in}

\begin{tabular}{c}
$^\dagger$University of California, Berkeley \\$^\ddagger$Stanford University
\end{tabular}

\vspace*{.2in}

\today

\vspace*{.2in}



%
%
\begin{abstract}
Due to its longevity and enormous information density,
DNA is an attractive medium for archival storage. 
In this work, we study the fundamental limits and tradeoffs of DNA-based storage systems
under a simple model, motivated by current technological constraints on DNA synthesis and sequencing.
Our model captures two key distinctive aspects of DNA storage systems:
(1) the data is written onto many short DNA molecules that are stored in 
an unordered way 
and 
(2) the data is read by randomly sampling from this DNA pool. 
Under this model, we characterize the storage capacity, and show that a simple index-based coding scheme is optimal.
\end{abstract}

\end{center}

\blfootnote{
$^*$Authors contributed equally and are listed alphabetically.
}

\vspace{-5mm}
\section{Introduction}
\label{sec:intro}




Recent years have seen an explosion in 
the demand for data storage,
posing significant challenges to current data centers and storage techniques.
This has spurred significant interest in new storage technologies beyond 
hard drives, magnetic tapes, and memory chips. 
In this context, 
DNA---the molecule that carries the genetic instructions of all living organisms---is a promising medium for archival data storage systems due to its longevity and enormous information density. 
%
%
%
%
The idea to store information on DNA is relatively new~\cite{baum_building_1995}, and has recently gained significant attention~\cite{church_next-generation_2012, goldman_towards_2013,grass_robust_2015,bornholt_dna-based_2016,erlich_dna_2016,yazdi_rewritable_2015}, due to advances in synthesis and sequencing. 
Several works \cite{church_next-generation_2012, goldman_towards_2013,grass_robust_2015,bornholt_dna-based_2016,erlich_dna_2016,yazdi_rewritable_2015} have demonstrated that writing, storing, and retrieving data on the order of megabytes
 is possible with today's technology, 
and achieves information densities and information lifetimes~\cite{grass_robust_2015} that are far beyond what state-of-the-art tapes and discs achieve. 
We refer to \cite{yazdi_survey_2015,erlich_dna_2016} for a survey and an overview of the recent advances in the area.

DNA is a long molecule made up of four nucleotides (Adenine, Cytosine, Guanine, and Thymine) and, for storage purposes, can be viewed as a string over a four-letter alphabet.
However, there are hard technological constraints for writing and reading DNA, which need to be considered in the design of a DNA storage system. 
While in a living cell DNA can consist of millions of bases, in practice it is difficult to synthesize very long strands of DNA.
In fact, all recently proposed systems 
\cite{church_next-generation_2012, goldman_towards_2013,grass_robust_2015,bornholt_dna-based_2016,erlich_dna_2016,yazdi_rewritable_2015} 
stored information on DNA molecules no longer than a few hundred nucleotides. 
Moreover, the process of determining the order of nucleotides in a DNA molecule, or sequencing, suffers from similar length constraints.
State-of-the-art sequencing platforms such as Illumina cannot sequence DNA segments longer than a few hundred nucleotides.
While recently developed ``third-generation'' technologies, such as 
Pacific Biosciences and Oxford Nanopore, can provide reads that are several thousand bases long, their error rates and reading costs are significantly higher.
Moreover, long reads are an overkill for storage systems in view of the current synthesizable lengths.




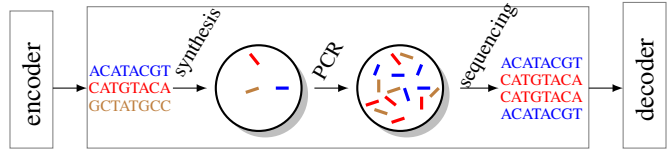
\begin{figure}[h!]
\begin{center}
\begin{tikzpicture}[scale=0.57,>=latex]
\draw[gray] (-2.8,-1.4) rectangle (-1.8,1.8);
\node at (-2.3,0.1) [rotate=90] {\small encoder};
\draw[->] (-1.8,0) -- (-1,0);
\draw[gray] (-1,-1.4) rectangle (10.7,1.9);
\draw[->] (10.7,0) -- (11.5,0);
\draw[gray] (11.5,-1.4) rectangle (12.5,2);
\node at (12,.2) [rotate=90] {\small decoder};

\node at (-0,0.4) {\tiny \textcolor{blue}{ACATACGT}};
\node at (-0,0) {\tiny \textcolor{red}{CATGTACA}};
\node at (-0,-0.4) {\tiny \textcolor{brown}{GCTATGCC}};

\draw[->] [->] (1,0) to (1.8,0);
\node at (1.5,0) [above,rotate=60] {\hspace{1.1cm}\scriptsize synthesis};

\begin{scope}[xshift=-10cm,yshift=3cm]
\draw[thick, drop shadow, fill=white] (13,-3) circle (1cm);
\draw[very thick, blue] (13.4,-3) to (13.7,-3);
\draw[very thick, red] (13,-2.45) to (12.8,-2.2);
\draw[very thick, brown] (12.7,-3.1) to (13.0,-3);
\end{scope}

\draw[->] (4.3,0) to node[above,swap,rotate=60] {\hspace{0.6cm}\scriptsize PCR
} (5.1,0);

\begin{scope}[xshift=-6.7cm,yshift=3cm]
\draw[thick, drop shadow, fill=white] (13,-3) circle (1cm);
\draw[very thick, blue] (13.1,-3) to (13.2,-3.3);
\draw[very thick, blue] (13.4,-3) to (13.7,-3);
\draw[very thick, blue] (13.4,-2.7) to (13.5,-2.45);
\draw[very thick, blue] (12.4,-2.7) to (12.5,-2.45);
\draw[very thick, blue] (12.8,-2.7) to (13.1,-2.7);
\draw[very thick, red] (12.8,-3.8) to (13.1,-3.7);
\draw[very thick, red] (12.2,-3.4) to (12.5,-3.3);
\draw[very thick, red] (12.9,-3.4) to (12.65,-3.2);
\draw[very thick, red] (13,-2.45) to (12.8,-2.2);
\draw[very thick, red] (13.5,-3.2) to (13.5,-3.5);
\draw[very thick, brown] (12.5,-3.1) to (12.5,-2.8);
\draw[very thick, brown] (12.7,-3.1) to (13.0,-3);
\draw[very thick, brown] (13.7,-3.2) to (13.9,-3);
\draw[very thick, brown] (13,-2.2) to (13.3,-2.3);
\draw[very thick, brown] (12.4,-3.5) to (12.7,-3.6);
\end{scope}

\draw[->] [->] (7.7,0) to (8.6,0);
\node at (8.3,0) [above,rotate=60] {\hspace{1cm}\scriptsize sequencing};

\begin{scope}[xshift=-1.1cm]
\node at (10.7,0.6) {\tiny \textcolor{blue}{ACATACGT}};
\node at (10.7,0.2) {\tiny \textcolor{red}{CATGTACA}};
\node at (10.7,-0.2) {\tiny \textcolor{red}{CATGTACA}};
\node at (10.7,-0.6) {\tiny \textcolor{blue}{ACATACGT}};
\end{scope}
\end{tikzpicture}
\end{center}
\caption{
\label{fig:channelmodel}
Channel model for DNA storage systems: the input to the channel is a multi-set of $\M$ length-$\len$ DNA molecules, while the output is a multi-set of $\N$ draws from the pool of DNA molecules.
}
\vspace{-2mm}
\end{figure}

Given these constraints, a simple model, illustrated in Fig.~\ref{fig:channelmodel}, that captures some of the main differences between a DNA storage system and conventional storage systems is as follows.
The data is written on $\M$ DNA molecules, each of length $\len$. 
Accessing the information is done via state-of-the-art sequencing technologies (including Illumina and third-generation sequencing technologies such as nanopore sequencing).
This corresponds to randomly sampling and reading molecules from the DNA pool. 
Since in practice sequencing is preceded by PCR amplification, which replicates each DNA molecule in the pool by roughly the same amount, we model the reading process as drawing $\N$ times uniformly at random, with replacement, from the $\M$ DNA molecules. 
%
The decoder's goal is to reconstruct the information from the set of $\N$ reads. 
Note that the decoder has no information about which molecules were sampled, and unless $\N = \Omega(\M \log \M)$, a proportion of the original DNA fragments is never sampled. 
Practical limitations of this model are discussed in the Concluding Remarks.


{\bf Contributions: }
In this paper we study the fundamental limits of the simple DNA storage model outlined above.
Our analysis aims to reveal the basic relationships and tradeoffs between key design parameters and performance goals such as storage density and reading/writing costs.
Motivated by the low error rate of Illumina sequencers and the existence of error-correcting codes tailored to the sequencing channel \cite{gabrys_asymmetric_2015}, we make the simplifying assumption that each sampled molecule is read in an error-free way.
We let the number of reads be $\N = \cNM \M$, and consider 
the asymptotic regime where $\M \to \infty$. 
The main parameter of interest is the storage capacity $C_s$, defined as the maximum number of bits that can be reliably stored per nucleotide (the total number of nucleotides is $\M\len$). 
Our main result shows that if $\lim_{\M \to \infty} \frac{\len}{\log \M} = \beta >1$, then
\[
C_s = (1 - e^{-\cNM}) (1- 1/\beta) \log(4).
\]
If $\lim_{\M \to \infty} \frac{\len}{\log \M} < 1$, no positive rate is achievable. 
The factor $1 - e^{-\cNM}$ can be  understood as the loss due to unseen molecules, and $1-1/\beta$ corresponds to the loss due to the unordered fashion of the reading process.


One important implication of this result is that a simple index-based scheme (as common in the literature) is optimal; i.e., prefixing each molecule with a unique header incurs no rate loss. 
More specifically, our result shows that indexing each DNA molecule and employing an optimal erasure code across the molecules
is information-theoretically optimal. 

Interestingly, as a bonus, the molecule's indices can be used as a primer to allow for random access via selectively sequencing of the DNA molecules, as proposed in \cite{yazdi_rewritable_2015}. 
Our result thus suggests that this index-based random access imposes no loss to the storage rate.
Furthermore, our expression for the capacity suggests that practical systems should not operate at a very high coverage depth $c = N/M$, as this significantly increases the time and cost of reading, and provides little storage gains. 
In particular, if $\M$ is large, it is wasteful to operate in the $\N = \Omega(\M \log \M )$ regime where each sequence is read at least once, as this only gives a marginally better storage capacity, while the sequencing costs can be exorbitant.

{\bf Related literature: }
The capacity of a DNA storage system under a related model has been studied in an unpublished manuscript by MacKay, Sayer, and Goldman~\cite{mackay_near-capacity_2015,sayir_challenge_2016}, which was brought to our attention after submission of the current work. 
In the model in~\cite{mackay_near-capacity_2015}, the input to the channel consists of a (potentially arbitrarily large) set of DNA molecules of fixed length $\len$, which is not allowed to contain duplicates. The output of the channel are $\M$ molecules drawn with replacement from that set. 
The approach in \cite{mackay_near-capacity_2015} considers coding over repeated independent storage experiments, and computes the single-letter mutual information over one storage experiment. This indicates that the price of not knowing the ordering of the molecules is logarithmic in the number of synthesized molecules, similar to our main result.

The capacity of a DNA storage system under a different model was studied in \cite{erlich_dna_2016}. 
Specifically \cite{erlich_dna_2016} assumes that each DNA segment is indexed which reduces the channel model to an erasure channel. 
While this assumption removes the key aspects that we focus on in this paper, namely that DNA molecules are stored in an unordered way and read via random sampling, \cite{erlich_dna_2016} considers other important constraints, such as homopolymer limitations.

Several recent works have designed coding schemes for DNA storage systems based on this general model, some of which were implemented in proof-of-concept storage systems \cite{church_next-generation_2012,goldman_towards_2013,grass_robust_2015,bornholt_dna-based_2016,erlich_dna_2016}. 
Several papers have studied important additional aspects of the design of a practical DNA storage system.
Some of these aspects include DNA synthesis constraints such as sequence composition  \cite{kiah_codes_2016,yazdi_rewritable_2015,erlich_dna_2016}, the asymmetric nature of the DNA sequencing error channel \cite{gabrys_asymmetric_2015}, 
the need for codes that correct insertion errors \cite{sala_insertions_2016}, 
and the need for techniques to allow random access \cite{yazdi_rewritable_2015}.
The use of fountain codes for DNA storage was considered in both \cite{erlich_dna_2016} and \cite{mackay_near-capacity_2015}. 

{\bf Outline: }
The paper is organized as follows.
After formalizing the problem setting in Section~\ref{sec:problem}, we prove the main result in Section~\ref{sec:storagecap}.
We then present the Storage-Recovery tradeoff that follows as a consequence in Section~\ref{sec:tradeoff}, and provide some discussion and future directions in Section~\ref{sec:discussion}.

\section{Problem Setting}
\label{sec:problem}


An $(\M,\len,\N)$ DNA storage code $\C$ is a set of codewords  
each of which is a multi-set of $M$ strings of length $\len$,
 together with a decoding procedure. 
The alphabet $\Sigma$ is typically $\{\bA,\bC,\bG,\bT\}$, corresponding to the four nucleotides that compose DNA. 
However, to simplify the exposition we focus on the binary case $\Sigma = \{0,1\}$, and we note that the results can be extended to a general alphabet in a straightforward manner. 
A codeword $\{\x_1,...,\x_M\} \in \C$ is stored as a physical mixture of $M$ synthesized DNA molecules of length $\len$. 
Throughout the paper we use the word molecule to refer to each of the stored strings of length $\len$ over the alphabet $\Sigma$. 
\begin{figure}[ht] 
	\center
       \includegraphics[width=6cm]{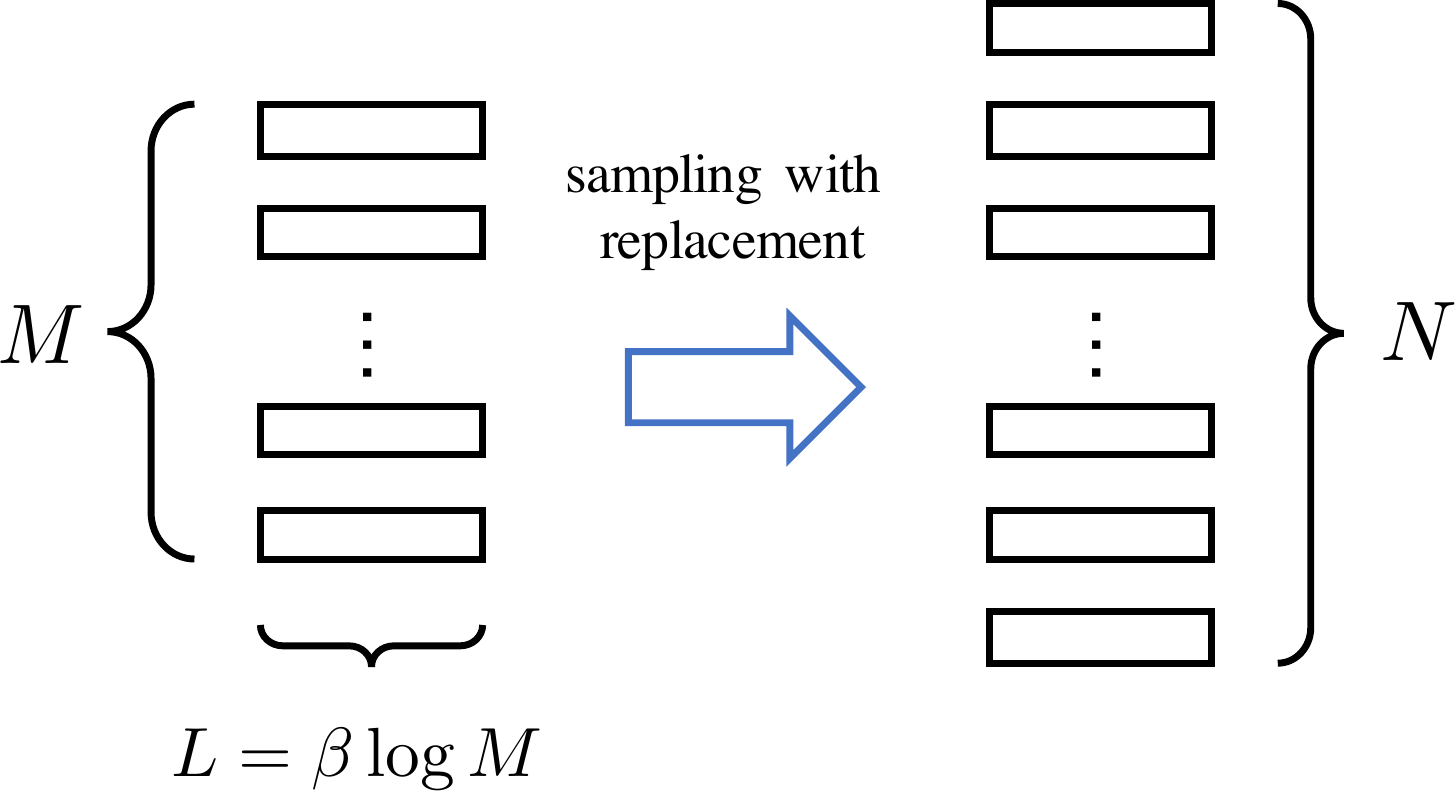} 
       \caption{Model for a DNA storage system.\label{fig:diagram}}
\end{figure}

In the recovery phase, the decoder draws $\N$ samples independently and uniformly at random, with replacement, from the set of stored molecules $\x_1,...,\x_\M$, obtaining the multi-set $\{\y_1,...,\y_N\}$.
A decoding function then maps $\{\y_1,...,\y_N\}$ to a message index in $\{1,\ldots,|\C|\}$.
We let $\cNM \defeq N/M$ be the coverage depth; i.e., the expected number of times each molecule $\x_i$ is sampled. 

The main parameter of interest of a DNA storage system is the storage density, or the storage rate, 
defined as the number of bits written per DNA base synthesized, i.e., 
\al{
R_s \defeq \frac{\log|\C|}{\M\len}. \label{eq:Rs}
}
Due to the nature of the reading process, via sequencing, another parameter of interest is the recovery rate, defined as the number of bits recovered per DNA base sequenced: 
\al{
R_r \defeq \frac{\log|\C|}{\N\len}. \label{eq:Rr}
}
We consider an asymptotic regime where the coverage depth $\cNM$ is fixed
and $M \to \infty$.
We will let $\len \defeq \beta \log M$ for some fixed $\beta $.
As our main result will show, $L = \Omega(\log M)$ is the asymptotic regime of interest for this problem.
We say that the rates $(R_s,R_r)$ are achievable if there exists a sequence of DNA storage codes $\C_\M$ with rates $(R_s,R_r)$ such that the decoding error probability tends to $0$ as $\M \to \infty$. 

\section{Storage Capacity}
\label{sec:storagecap}


Our main result is the characterization of the storage capacity, given by the following theorem.

\begin{theorem} \label{thm:storagecap}
The storage rate $R_s$ is achievable if and only if 
\al{
R_s  & \leq (1-e^{-\cNM})\left(1 - 1/\beta \right). \label{eq:storagecap}
}
In particular, if $\beta \leq 1$, no positive rate is achievable.
\end{theorem}

%




The capacity expression in (\ref{eq:storagecap}) can be intuitively understood through the achievability argument.
A storage rate of $R_s = (1-e^{-\cNM})\left(1 - 1/\beta \right)$ can be easily achieved by prefixing all the molecules with a distinct tag, which effectively converts the channel to an erasure channel. 
More precisely, we use the first $\log \M$ bits of each molecule to encode a distinct tag. 
Then we have $\len - \log \M = \len(1 - 1/\beta)$ symbols left per molecule to encode data. 
The decoder can use the tags to remove duplicates and sort the molecules that are sampled. 
This effectively creates an erasure channel, where a molecule is erased if it is not sampled in any of the $N$ attempts, which occurs with probability 
 $(1 - 1/\M)^{\M \cNM}$.
Since $\lim_{\M\to\infty} (1 - 1/\M)^{\M \cNM} = e^{-\cNM}$, the storage rate $R_s$ achieved is $(1-e^{-\cNM})(1-1/\beta)$.
The surprising aspect of Theorem~\ref{thm:storagecap} is that this simple index-based scheme is optimal.


\subsection{
\label{sec:motivationconverse}
Motivation for Converse}

A simple outer bound can be obtained by considering a genie that provides the decoder with the index of 
each sampled molecule.
In other words, $\{\x_1,\ldots,\x_\M\}$ are the stored molecules, and the decoder observes $\{ (\y_1,i_1),(\y_2,i_2),\ldots,(\y_\N,i_\N) \}$ where $i_j$ is such that $\y_j = \x_{i_j}$.
This converts the channel into an erasure channel with erasure probability $(1 - 1/\M)^{\M \cNM}$, 
and taking the limit, 
\al{
R_s \leq 1 - e^{-\cNM} \label{eq:simplebound1}
}
follows.
It is intuitive that the bound~\eqref{eq:simplebound1} should not be achievable, as the decoder in general cannot sort the molecules and create an effective erasure channel.
However, it is not clear a priori that prefixing every molecule with an index is 
optimal.

Notice that one can view the DNA storage channel as a channel where
the encoder chooses a distribution (or a type) over the alphabet $\Sigma^\len$
and the decoder takes $N$ samples from this distribution.
From this angle, the question becomes ``how many types $\type \in \Z_+^{2^L}$ with $\|\type \|_1 = M$ can be reliably decoded from $N$ independent samples?'', and restricting ourselves to index-based schemes restricts the set of types to those with $\|\type \|_\infty = 1$; i.e., no duplicate molecules are stored.

While this restriction may seem suboptimal, a counting argument suggests that it is not.
The number of types for a sequence of length $M$ over an alphabet of size $|\Sigma^\len| = 2^L$ is at most $M^{2^L}$ and thus at most
\aln{
\frac{1}{ML} \log M^{2^L} = \frac{2^L \log M}{M \beta \log M}  = \frac{2^L}{\beta M} 
}
bits can be encoded per symbol.
We conclude that, if $\lim_{M \to \infty} L/\log M < 1$, $R_s = 0$.
An actual bound on $R_s$ can be obtained by counting the number of types more carefully.
This is done in the following lemma, which we prove in the appendix.
\begin{lemma} \label{lem:comb}
The number of distinct vectors $\type \in \Z_+^{a}$ with $\|\type \|_1 = b$ is given by
\aln{
\T[a,b] \defeq {a+b-1 \choose b} < \left( \frac{e(a+b-1)}{b} \right)^b.
}
\end{lemma}
Since our types are vectors $\type \in \Z_+^{2^L}$ with $\|\type \|_1 = M$, and $2^L  = 2^{\beta \log M} = M^\beta$, it follows that at most 
\aln{
\frac{1}{ML} \log \left( \frac{e(M^\beta+M-1)}{M} \right)^{M} \leq \frac{M \log \alpha M^{\beta-1}}{M \beta \log M}
}
bits can be encoded per symbol, for some $\alpha > 1$, and 
\al{
R_s \leq 1- 1/\beta.  \label{eq:simplebound2}
}
Therefore, if we had a deterministic channel where the decoder  observed \emph{exactly} the $M$ stored molecules, an index-based approach would be optimal from a rate standpoint.
The converse presented in the next section utilizes a more careful genie to show
that the bounds in  (\ref{eq:simplebound1}) and (\ref{eq:simplebound2}) can in fact be combined, implying the optimality of index-based coding approaches.

\subsection{Converse}

\newcommand\set{\mathrm{set}}

Let $\x_i$, $i=1,\ldots,\M$, be the $\M$ length-$\len$ molecules written into the channel 
and $\y_i$, $i=1,\ldots,\N  = \cNM \M$, be the length-$\len$ molecules observed by the decoder. 
Notice that, since the sampling is done with replacement, whenever the channel output is such that $\y_i = \y_j$ for $i \ne j$, the decoder cannot determine whether both $\y_i$ and $\y_j$ were sampled from the same molecule $\x_\ell$ or from two different molecules that obey $\x_\ell = \x_k, \ell \neq k$. 
In order to derive the converse, we consider a genie-aided channel that removes this ambiguity. 
\begin{figure}[h] 
\vspace{0mm}
	\center
       \includegraphics[width=9cm]{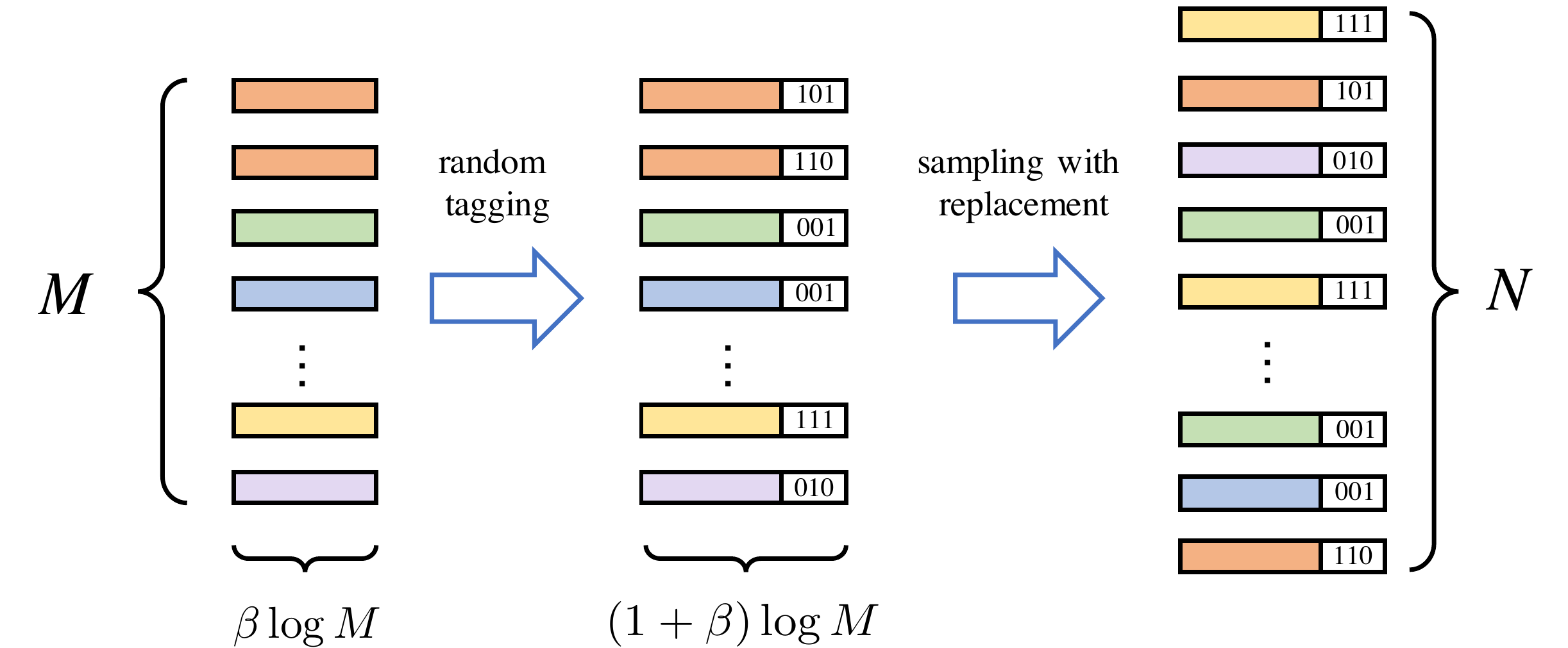} 
       \caption{Genie-aided channel for converse.\label{fig:geniechannel}}
\end{figure}
As illustrated in Fig.~\ref{fig:geniechannel}, before sampling the $\N$ molecules, the genie-aided channel appends a unique index of length $\log \M$ to each molecule $\x_i$, which results in the set of tagged molecules $\{\tilde \x_i\}_{i=1}^\M$. 
We emphasize that those indices are all unique, and are chosen randomly and independently from the input sequences $\{\x_i\}_{i=1}^\M$. 
Notice that, in contrast to the naive genie discussed in Section~\ref{sec:motivationconverse}, this genie does \emph{not} reveal 
the index $i$ of the molecule $\x_i$ from which $\y_\ell$ was sampled.
Therefore, the channel is \emph{not} reduced to an erasure channel, 
and intuitively the indices are only useful for the decoder to determine whether two equal samples $\y_\ell = \y_k$ came from the same molecule or from distinct molecules. 

The output of the genie-aided channel, denoted by $\{\tilde \y_i \}_{i = 1}^\N$, is then obtained by sampling $\N$ times with replacement from the set of tagged molecules $\{\tilde \x_i\}_{i=1}^\M$. 
It is clear that any storage rate $R_s$ achievable in the original channel can be achieved on the genie-aided channel, as the decoder can simply discard the indices, or stated differently, the output of the original channel can be obtained from the output of the genie-aided channel. 

Next, let $\set( \{\tilde \y_i \}_{i = 1}^\N )$ be the set obtained from $\{\tilde \y_i \}_{i = 1}^\N$ by removing any duplicates. 
Then $\set( \{\tilde \y_i \}_{i = 1}^\N )$ is a sufficient statistic for the parameter $\{\x_i\}_{i=1}^\M$ since all tagged molecules are distinct objects, and sampling the same molecule $\tilde \x_i$ does not yield additional information on $\{\x_i\}_{i=1}^\M$. 
More formally, conditioned on $\set( \{\tilde \y_i \}_{i = 1}^\N )$, $\{\x_i\}_{i=1}^\M$ is independent of the genie's channel output $\{\tilde \y_i \}_{i = 1}^\N$. 

Next, we define the frequency vector $\vf \in \Z_+^{\M^{\beta}}$ 
(note that $|\Sigma^{\len}| = 2^{\beta \log \M} = \M^\beta$) 
obtained from $\set( \{\tilde \y_i \}_{i = 1}^\N )$ in the following way. 
The $\y$-th entry of $\vf$, $\y \in \Sigma^{\len}$, is given by 
\aln{
\vf[\y] 
\defeq \left| \left\{ \tilde \y \in \set( \{\tilde \y_i \}_{i = 1}^\N) 
\colon  \tilde \y[1\colon \len] = \y \right\} \right|. 
}
Since $\set( \{\tilde \y_i \}_{i = 1}^\N)$ is a sufficient statistic for $\{\x_i\}_{i=1}^\M$
and the tags added by the genie were chosen at random and independently of $\{\x_i\}_{i=1}^\M$, it follows that $\vf$ is also a sufficient statistic for $\{\x_i\}_{i=1}^\M$. 
Hence, we can view the (random) frequency vector $\vf$ as as the output of the channel without any loss. 

Notice that $|\set( \{\tilde \y_i \}_{i = 1}^\N )| =  \| \vf \|_1$, and in expectation we have 
$\EX{ \| \vf \|_1 / \M } = (1 - (1 - 1/\M)^{\M \cNM})$ 
which converges to
$1 - e^{-\cNM}$ for $\M \to \infty$. 
Furthermore, the following lemma, which we prove in the appendix, asserts that $ \| \vf \|_1$ does not exceed its expectation by much.

\begin{lemma} \label{lem:conc}
For any $\delta > 0$, the frequency vector $\vf$ at the output of the genie-aided channel satisfies
\aln{
\Pr\left( \frac{\| \vf \|_1}{\M} > 1-e^{-\cNM} + \delta \right) \to 0, \text{ as } \M \to \infty.
}
\end{lemma}


We now append the coordinate  
$
F_0 = (1-e^{-\cNM}+ \delta) \M - \| \vf \|_1
$
to the beginning of $\vf$ to construct $\vf' = (F_0,\vf)$.
Notice that when $\| \vf \|_1 \leq (1-e^{-\cNM} + \delta)\M$ (which by Lemma~\ref{lem:conc} happens with high probability), we have $\| \vf' \|_1 = (1-e^{-\cNM} + \delta)\M$.

Fix $\delta > 0 $, and define the event $\E =  \{ \| \vf \|_1 > (1-e^{-\cNM} + \delta)\M \}$ with indicator function $\one_\E$.
By Lemma~\ref{lem:conc}, $\PR{\E} \to 0$ as $\M \to \infty$. 
Consider a sequence of codes $\{\C_\M\}$ with rate $R_s$ and vanishing error probability.
If we let $W$ be the message to be encoded, chosen uniformly at random from $\{1,\ldots,2^{\M \len R_s}\}$,
from Fano's inequality we have
\al{
\M \len R_s  
&= 
H(W)
=
I(W; \vf') + H(W| \vf') \nonumber \\
&\leq H(\vf') + 1 + P_e \M \len R_s,
}
where $P_e$ is the probability of a decoding error, which by assumption goes to zero as $M \to \infty$. 
We can then
upper bound the achievable storage rate $R_s$ as
\al{
\M & \len R_s (1-P_e) 
\leq H (\vf' )  +1 \leq  H\left( \vf', \one_\E \right) +1 \nonumber \\
&\leq \PR{ \E } H\left( \left. \vf'\, \right|  \E  \right) + \PR{ \bar \E } H \left( \left. \vf' \, \right|  \bar \E  \right) + H(\one_\E) + 1 \nonumber \\
&\leq \PR{ \E } \log \T [\M^\beta+1,\M] \nonumber \\ 
&\quad \quad + \log \T [\M^\beta+1, (1-e^{-\cNM} + \delta)\M] + 2, \label{eq:rsbound}
}
where $\T[a,b]$ is the number of vectors $x \in \Z_+^a$ with $\| x \|_1 = b$.
An application of Lemma~\ref{lem:comb} yields 
\aln{
\log \T & [\M^\beta+1, (1-e^{-\cNM} + \delta)\M] \\
& \leq (1-e^{-\cNM} + \delta)\M \log \left( e + \frac{e M^{\beta-1}}{(1-e^{-\cNM} + \delta)} \right) \\ 
& \leq (1-e^{-\cNM} + \delta)\M \log \left( \alpha M^{\beta-1} \right) \\
& \leq (1-e^{-\cNM} + \delta)\M  [ (\beta-1) \log \M + \log \alpha ],
}
where $\alpha$ is a positive constant. 
Analogously, we obtain 
\aln{
\log \T [\M^\beta+1,\M]  \leq \M ( (\beta -1) \log \M + \log \alpha ). 
}
Dividing (\ref{eq:rsbound}) by $\M\len$ and applying the bounds above yields
\aln{
R_s (1 - P_e) &\leq \Pr( \E ) \frac{ \M [(\beta -1) \log \M + \log \alpha ]}{\M \len} \\
& -\hspace{-10mm} + \frac{( 1 - e^{-\cNM} + \delta)\M  [ (\beta-1) \log \M + \log \alpha ]}{\M \len} + \frac{2}{\M\len} \\
& \leq \Pr( \E ) \left( \frac{\beta - 1}{\beta} +\frac{\log \alpha}{\beta \log M}  \right) \\ 
& -\hspace{-10mm} + (1-e^{-\cNM} + \delta) \left(1-\frac{1}{\beta}  + \frac{\log \alpha}{\beta \log \M}\right)  + \frac{2}{\M\len}.
}
Finally, letting $\M \to \infty$ yields
\aln{
R_s \leq (1-e^{-\cNM} + \delta) \left(1- 1/\beta \right),
}
since $\Pr(\E) \to 0$ by Lemma~\ref{lem:conc}.

\section{Storage-Recovery Tradeoff} 
\label{sec:tradeoff}

Most studies on DNA-based storage emphasize the storage rate (or storage density), while sequencing  costs are disregarded.
From a practical point of view, it is important to understand, for a given storage rate, how much sequencing is required for reliable decoding, as this determines the time and cost required for retrieving the data.
Thus, characterizing the set of pareto-optimal points, or the optimal storage-recovery tradeoff, is of relevance.
From Theorem~\ref{thm:storagecap} and 
the fact that $R_s = \cNM R_r$, the $(R_s,R_r)$ feasibility region is as follows:

\begin{cor}
The rates $(R_s,R_r)$ are achievable if and only if, for some $c > 0$,
\al{
R_s  & \leq (1-e^{-\cNM})\left(1 - 1/\beta \right),\nonumber \\
R_r & \leq \frac{1-e^{-\cNM}}{\cNM}\left(1 - 1/\beta \right). 
\nonumber
}
\end{cor}


This region is illustrated in Fig.~\ref{fig:region}.
This tradeoff suggests that a good operating point would be achieved by not trying to maximize the storage rate (which technically requires $\cNM \to \infty$).
Instead, by using some modest coverage depth $\cNM=1,2,3$, most of the storage rate ($63 \%,86\%,95\%$, respectively) can be achieved.
This is somewhat in contrast to what has been done in the practical DNA storage systems that have been developed thus far, where the decoding phase utilizes very deep sequencing.

To be concrete, suppose that we are interested in minimizing the cost for storing data on DNA. 
Synthesis costs are currently larger than sequencing costs by about a factor $q = 10,000$-$100,000$. 
Thus, if our goal is to minimize the cost for synthesizing and sequencing a given number of bits in DNA, the cost is proportional to $q/R_s + 1/R_r = \frac{q+\cNM}{1-e^{-\cNM}}$. 
This quantity can be maximized over $\cNM$, to obtain the optimal cost per bit stored.
For example, for $q=10000$, $\cNM \approx 9.2$. 
Moreover, one might be interested in optimizing other quantities such as reading time or considering a scenario where the data is read more than once.

\begin{figure} [h]
	\center
       \includegraphics[width=5.6cm]{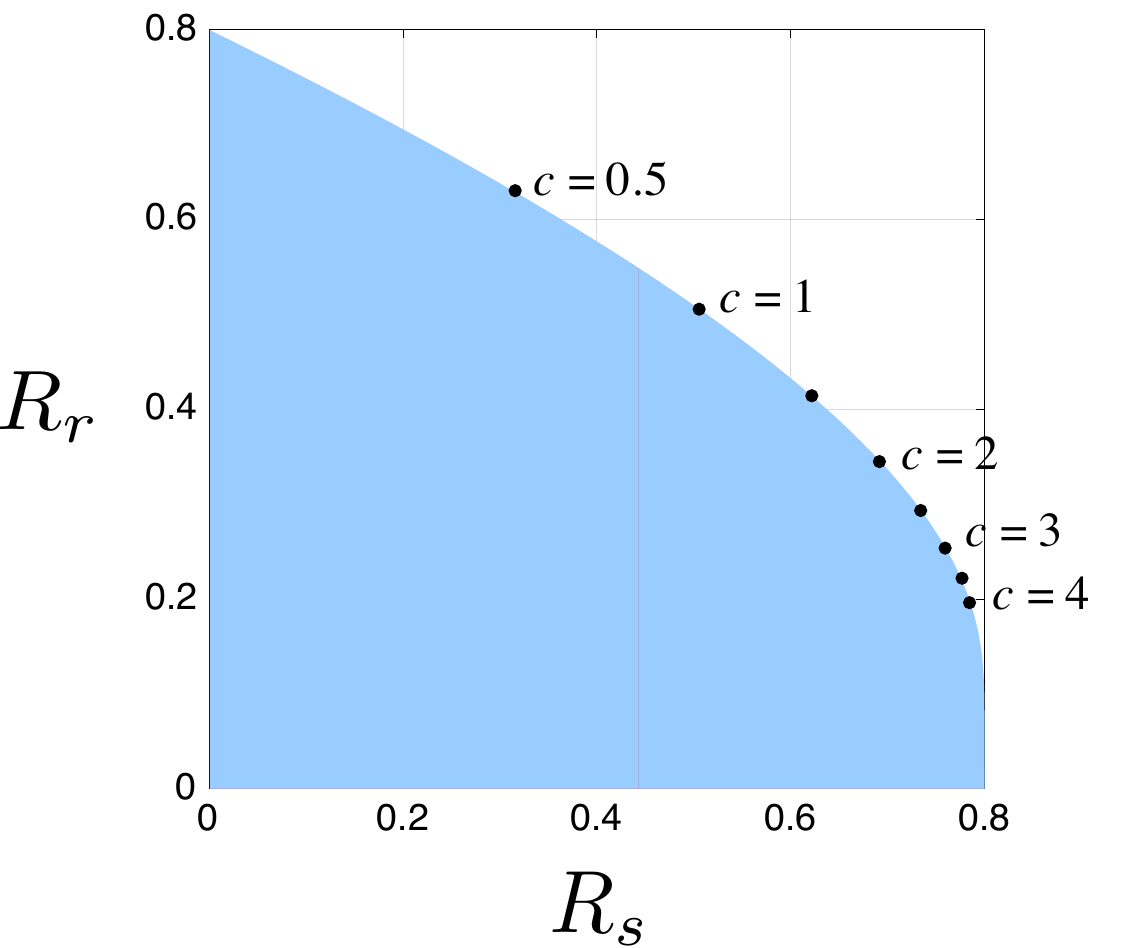} 
       \caption{$(R_s,R_r)$ feasibility region for $\beta = 5$. \label{fig:region}}
\end{figure}

\section{Concluding Remarks}
\label{sec:discussion}

In this paper we took initial steps in the study of the fundamental limits of DNA-based storage systems.
We proposed a simple model inspired by the current technological landscape of DNA synthesis and sequencing technologies.
Under this model, we showed that a simple index-based coding scheme is asymptotically optimal.

{\bf Model limitations: }
While our model captures key aspects of the writing and reading processes of DNA-based storage systems, several aspects are left out.
For example, the assumption that each molecule is written and read error-free, and is not damaged during storage is not realistic. 
Even though error-correcting codes can be used at the molecule block level, it is not clear whether a ``separation-based'' approach (with an inner and an outer code) is optimal. 
Moreover, the performance of the index-based coding scheme that achieves the storage capacity could be significantly affected, since errors in the indices would impact the ordering of the blocks, and the channel would no longer be effectively reduced to an erasure channel.

Moreover, the assumption that PCR amplifies all molecules by the same amount is also not realistic, in particular when the number of PCR cycles is large.
Hence, a more realistic model could assume the distribution of the molecules in the DNA pool to be non-uniform (for example Poisson).




{\bf Future work: }
Natural extensions of this work are to address the model limitations described above.
Another direction for further investigation relates to the asymptotic regime we considered.
Although our main results states that no rate is achievable when $\beta \leq 1$,
in practice we can store significant amounts of data even with very short DNA molecules. 
Hence, formulating a reasonable model to study the very-short-molecule regime would be of interest.
The regime of long molecules may also become practically relevant as synthesis technologies evolve.





\section*{Funding and Acknowledgements}
RH was supported by the Swiss National Science Foundation under grant P2EZP2\_159065.
IS and DT were supported by the Center for Science of Information, an NSF
Science and Technology Center, under grant agreement CCF-0939370.

Part of this work was completed while IS and DT were visiting the Simons Institute for the Theory of Computing, UC Berkeley, and IS was supported by a Simons Research Fellowship.
The authors would like to thank Olgica Milenkovic for helpful discussions on DNA-based storage during that time. 

Moreover, the authors thank Yaniv Erlich very much for helpful comments on a previous version of this manuscript, and Jossy Sayir for bringing our attention to the manuscript~\cite{mackay_near-capacity_2015}. 
Finally, RH would like to thank Robert Grass for helpful discussions. 



\printbibliography

\section{Appendix}

\subsection{Proof of Lemma~\ref{lem:comb} }
Notice that vectors  $x \in \Z_+^{a}$ with $\|x \|_1 = b$ are in one-to-one correspondence with 
binary strings containing $(a-1)$ $0$s and $b$ $1$s.
For $x = (x_1,...,x_a)$,
the corresponding string is 
\al{
\underbrace{1 \, \ldots \, 1}_{x_1} \, 0 \, \underbrace{1 \, \ldots \, 1}_{x_2} \, 0 \, \ldots \, 0 
\underbrace{1 \, \ldots \, 1}_{x_a}.
}
It is clear that such a string has $(a-1)$ $0$s and $b$ $1$s,
and that distinct strings with $(a-1)$ $0$s and $b$ $1$s correspond to distinct vectors $x$.
The number of distinct strings of this form is
\aln{
\frac{(a-1+b)!}{(a-1)!\, b!} = {a+b-1 \choose b}.
}
The upper bound in the statement of the lemma is a standard bound for binomial coefficients.


\subsection{Proof of Lemma~\ref{lem:conc}} 
The $\ell_1$ norm of the frequency vector $\vf$ at the output of the genie-aided channel is distributed as the number of distinct coupons obtained by drawing $\N = \cNM \M$ times with replacement from a set of $\M$ distinct coupons. 
Thus, Lemma~\ref{lem:conc} is an immediate consequence of the following stronger statement. 

\begin{lemma}
Let $Q$ be the number of distinct coupons obtained by drawing $\N = \cNM \M$ times with replacement from a set of $\M$ distinct coupons. 
We have that, for any $\delta > 0$,
\begin{align*}
\PR{Q \geq (1- e^{-\cNM} + \delta) \M}
\leq
\frac{1}{\M} \frac{2e^{2\cNM}}{2\left( \ln \left( \frac{e^{-\cNM}}{e^{-\cNM} - \delta} \right) - \frac{e^{\cNM}}{\M}\right)^2}.
\end{align*}
\end{lemma}

\begin{proof}
Since $\PR{Q \geq (1- e^{-\cNM} + \delta) \M}$ is a non-increasing function of $\delta$, we can assume that $\delta \in (0, e^{-\cNM}/2]$, as that simplifies the expressions.
Let $t_i$ be the number of draws to collect the $i$-th coupon after $(i-1)$ coupons have been collected, $i=0,\ldots,\M-1$, and consider the number of draws for obtaining $\alpha\M$ distinct coupons 
$
T \defeq \sum_{i=0}^{\alpha \M-1} t_i$
where
$\alpha \defeq 1 - e^{-\cNM} + \delta.
$
Due to
\[
\PR{Q \geq (1- e^{-\cNM} + \delta) \M}
=
\PR{Q \geq \alpha \M}
= 
\PR{T \leq \N}, 
\]
the lemma will follow by upper-bounding $\PR{T \leq \N}$ using Chebyshev's inequality. 
We first note that with 
$\EX{t_i} = 1/p_i, p_i \defeq \frac{\M-i}{\M}$ and
$\Var{t_i} = \frac{1-p_i}{p_i^2}$, 
we obtain 
\begin{align}
\EX{T} &
=
\sum_{i=0}^{\alpha \M-1} \EX{t_i}
= 
\M \sum_{i=0}^{\alpha \M-1} \frac{1}{\M-i}  \nonumber \\
&= 
\M(H_{\M} - H_{\M (1-\alpha) }) \nonumber \\
&\geq
\M( \ln \M - \ln (\M (1-\alpha)) ) - \frac{1}{2(1-\alpha)} \nonumber  \\
&\geq
-\M \ln(1-\alpha) - e^{\cNM}
= -\M \ln(e^{-\cNM} - \delta) - e^{\cNM} \nonumber \\
&= 
\M \cNM 
+ \M \underbrace{
\ln \left( \frac{e^{-\cNM}}{e^{-\cNM} - \delta} \right)
}_{\xi} - e^{\cNM}
= \N + \M \xi - e^{\cNM}. \nonumber
\end{align}
Here, $H_\M = \sum_{i=1}^\M \frac{1}{i}$ is the $\M$-th harmonic number, and the first inequality follows by the asymptotic expansion
\[
0 \leq H_n - \ln n - \gamma 
= 
\frac{1}{2n}  - \frac{1}{12 n^2} +  \frac{1}{120 n^4}  - \ldots 
\leq \frac{1}{2n},
\]
where $\gamma$ is the Euler-Mascheroni constant. 
The second inequality follows from $\frac{1}{1 - \alpha} \leq \frac{1}{e^{-\cNM} - e^{-\cNM}/2 } = 2 e^{\cNM}$. 
Moreover, the variance can be upper-bounded as
\begin{align}
\Var{T}
&=
\sum_{i=0}^{\alpha \M-1} \Var{t_i}
=
\sum_{i=0}^{\alpha \M-1} \frac{i\M}{(\M-i)^2} \nonumber \\
&
\leq 
\M \frac{\alpha}{2 (1-\alpha)^2}
\leq 
\M 2e^{2\cNM}.
\label{eq:boundvar}
\end{align}
Using the bound on the expectation and Chebyshev's inequality, we have for any $\beta>0$, that
\begin{align}
&\PR{ -T  + \N + \M \xi - e^{\cNM}  > \beta } \nonumber \\
&\hspace{2cm}\leq
\PR{ -T  + \EX{T}  > \beta }
\leq \frac{\Var{T}}{\beta^2}. \nonumber
\end{align}
Choosing $\beta = \M \xi - e^{\cNM}$ and using the upper bound on $\Var{T}$ given in~\eqref{eq:boundvar}, yields 
$
\PR{T \leq \N}
%
%
\leq 
\frac{1}{\M} \frac{2e^{2\cNM}}{\left(\xi - \frac{e^{\cNM}}{\M}\right)^2},
$
which concludes the proof. 
\end{proof}



\end{document}